\documentclass[letterpaper, 10 pt, conference]{ieeeconf}

\IEEEoverridecommandlockouts                              

\usepackage{graphics} 
\usepackage{epsfig} 
\usepackage{mathptmx} 
\usepackage{times} 
\usepackage{amsmath} 
\usepackage{amssymb}  
\usepackage{hyperref}
\usepackage{color, soul}
\usepackage{url}
\usepackage{cases}
\usepackage{booktabs}
\usepackage{multirow}
\usepackage{caption}
\usepackage{subcaption}
\usepackage{cite}

\newtheorem{theorem}{Theorem}
\newtheorem{lemma}{Lemma}
\newtheorem{corollary}{Corollary}
\newtheorem{assumption}{Assumption}

\title{\LARGE \bf
Hybrid SIS Dynamics for Demand Modeling \\ of Frequently Updated Products
}

\author{Ian Walter$^{1}$, Jitesh H. Panchal$^{2}$, and Philip E. Par\'e$^{1}$
\thanks{*This work was partially supported by the National Science Foundation, grant NSF-ECCS \#2238388.}
\thanks{$^{1}$Ian Walter and Philip E. Par\'e are with the Elmore Family School of Electrical and Computer Engineering, Purdue University, West Lafayette, IN 47907, USA 
        {\tt\small walteri@purdue.edu, phil.pare@purdue.edu}}%
\thanks{$^{2}$Jitesh H. Panchal is with the School of Mechanical Engineering, Purdue University, West Lafayette, IN 47907, USA
        {\tt\small panchal@purdue.edu}}%
}

\begin{document}

\maketitle
\thispagestyle{empty}
\pagestyle{empty}

\begin{abstract}
We propose a hybrid spreading process model to capture the dynamics of demand for software-based products. We introduce discontinuous jumps in the state to model sudden surges in demand that can be seen immediately after a product update is released. After each update, the modeled demand evolves according to a continuous-time susceptible-infected-susceptible (SIS) epidemic model. We identify the necessary and sufficient conditions for estimating the hybrid model's parameters for an arbitrary finite number of sequential updates. We verify the parameter estimation conditions in simulation, and evaluate how the estimation of these parameters is impacted by the presence of observation and process noise. We then validate our model by applying our estimation method to daily user engagement data for a regularly updating software product, the live-service video game `Apex Legends.' 
\end{abstract}

\section{Introduction}

In recent years, the rapid growth of digital products, particularly in sectors such as video games and software services, has emphasized the need for effective models to capture and predict changes in user demand.
A prominent feature of many modern software systems, especially live-service products, is the release of regular updates aimed at improving functionality, fixing issues, or adding new content. 
These updates often result in discontinuous jumps in demand, driven by heightened user interest and engagement immediately following an update's release. 
Understanding and modeling this behavior is crucial for developers to effectively plan updates, allocate resources, and enhance the user's experience.

One approach for capturing this behavior is epidemic modeling~\cite{doi:10.1177/1075547095016003002, rogers2010diffusion, doi:10.1073/pnas.1004098107, 8896074, 10.1007/s10588-006-9007-2, 10.1007/s10100-011-0210-y, 4549746, 6782297}.
How innovations, or products, spread has been studied since the works of Tarde~\cite{de1903laws}.
How innovations spread over social networks has been an area of particular emphasis in the last two decades~\cite{10.1007/s10100-011-0210-y, 10.1007/s10588-006-9007-2, doi:10.1073/pnas.1004098107, 4549746, 6782297}.
For instance, She \textit{et al.}\ explored how antagonistic relationships between individuals influence innovation diffusion in a network~\cite{9693279}.
To analyze the diffusion of innovation in large populations, researchers often employ mean-field approximations of networked diffusion processes
~\cite{GUSEO2009806}.

In order to capture the discontinuous effect that product updates can have on the demand, we leverage a hybrid system model which incorporates spreading process dynamics.
Liu \textit{et al.} present a switched system model implementing 
SEIR model dynamics~\cite{liu2010lasalle}.
The SEIR model does not allow for an endemic equilibrium to be reached, which means that the infected population must die out.
Hence, the SEIR model is not suitable to model the dynamics of a system that should persist, such as demand for live-service products.
Additionally, their proposed model can only change the dynamics at each discrete event.
While an update to a product may change the model parameters, it can also induce an impulse of demand, instantaneously changing the state.

We present a novel hybrid epidemic model leveraging continuous-time SIS dynamics and discrete impulse events, which cause the state and spreading parameters to change.
A discrete-time approximation of the hybrid model is then introduced, to enable parameter estimation.
We then present an estimation technique and analyze the necessary and sufficient conditions.
The estimation technique is tested on synthetic data generated with the hybrid model, and its sensitivity to noise is evaluated.
We leverage the estimation technique to validate the model 
by estimating parameters and predicting the future state of real data, which consists of daily user counts for an online game.

The rest of the paper is organized as follows.
In Section~\ref{sec:model}, we present the novel hybrid model and its discrete-time approximation. 
The necessary and sufficient conditions for estimating the model parameters are analyzed in Section~\ref{sec:main_results}.
The user-count dataset we use to evaluate the practical application of the model is presented in Section~\ref{sec:data}.
Simulations verifying the estimation technique are presented in Section~\ref{sec:simulation}, followed by an application of this system to demand fitting for our dataset.
This paper concludes with some remarks in Section~\ref{sec:conclusion}.

\subsection{Notation}

For any positive integer $n$, we use $[n]$ to denote the set $\{1,2,\dots,n\}$.
We use $x'$ to denote the transpose of a vector $x$ and $A'$ to denote the transpose of a matrix $A$.
The $i$th entry of a vector $x$ is denoted by $x_i$, and the $ij$th entry of matrix $A$ will be denoted by $a_{ij}$.
We define $t^-$ as $\sup \{t^-\in \mathbb{R} : t^-<t\}$. 
We use $\lceil{x}\rceil$ to denote the ceiling function, which returns $\min\{y\in\mathbb{N}:y\geq x 
\}$ where $\mathbb{N}$ is the set of natural numbers, including zero. 




\section{Model}\label{sec:model}
We propose the following hybrid model. 
Consider a population of $N$ individuals who are potential users of a product.
Assume that users can freely start and stop using the product at any time.
The number of currently active users is given by the state $I\in\{0\}\bigcup [N]$, and the proportion of currently active users is given by the state $x=I/N$.
Assume the product features and customer preferences change only at discrete times, called update times.
We denote the $i$th update time by $t_i$. 
We assume there are a total of $m$ updates.
Each update occurs at some time after the previous update, that is, $0<t_1<\dots<t_m$.
We use $t_0$ to denote the initial time and $t_{m+1}$ to denote the final time.
After each update, $x$ evolves based on the spreading process dynamics
\begin{subequations}\label{eq:ct_hybrid_dynamics}
\begin{align}\label{eq:ct_sis_dynamics}
    \dot{x} = \beta_i(1-x)x  - \gamma_i x,
\end{align}
where $\beta_i$ is the engagement (infection) rate and $\gamma_i$ is the disinterest (healing) rate corresponding to interval $i\in\{0\}\bigcup [m+1]$. 
At the time of an update, $t\in\{t_1,t_2,\dots,t_m\}$, there is a discontinuous shift in user engagement according to the rule
\begin{align}\label{eq:ct_alpha_dynamics}
    x(t) = (1+\alpha_i)x(t^-) ,
\end{align}
\end{subequations}
where 
$\alpha_i$ is the scaling factor for $i\in [m]$. 
An example graph representation of this hybrid model is shown in Figure~\ref{fig:system-graph} for $m=2$ updates.

\begin{figure}[t]
    \centering
    \includegraphics[width=0.9\columnwidth]{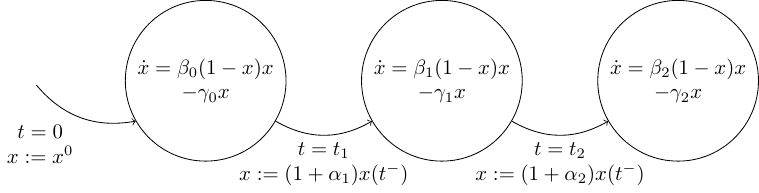}
    \caption{Example system graph for a system with two update events occurring at times $t_1$ and $t_2$. 
    The system starts in the first (left) state at time $t_0=0$ with initial state $x^0$, with the continuous dynamics  $\dot{x}=\beta_0(1-x)x - \gamma_0 x$. At time $t=t_1$ the value of state $x$ instantaneously changes to $(1+\alpha_1)x^-$, where $x^-=x(t^-)$. 
    The dynamics also change to use the new parameters $\beta_1, \gamma_1$.
    At time $t=t_2$, the value of state $x$ is set to $(1+\alpha_2)x^-$, and the state dynamics change to use parameters $\beta_2, \gamma_2$.}
    \label{fig:system-graph}
\end{figure}



All data is sampled, including the user data in this paper; see Section~\ref{sec:data}.
Therefore, we define a discrete-time model approximation of the above model for the parameter estimation and model validation.
We denote the time-step of the $i$th update by $T_i$, where $T_i=\lceil{t_i/h}\rceil$. 
We assume there are a total of $m$ updates, and use $T_0$ and $T_{m+1}$ to denote the initial and final time steps, respectively.
Each update occurs at some time after the previous update, that is, $0<T_1<\dots<T_m$.
The set of update times is defined as $\mathcal{T}=\{T_1,\dots,T_m\}$. 
We obtain the discrete-time approximation of the process dynamics following each update event by applying the Euler method~\cite{atkinson1991introduction} 
to~\eqref{eq:ct_sis_dynamics}:
\begin{subequations}\label{eq:dt_hybrid_dynamics}
\begin{align}\label{eq:dt_sis_dynamics}
    x^{k+1} = x^k+h(\beta_i(1-x^k)x^k  - \gamma_i x^k).
\end{align}
At the time-step preceding an update, $k=T_i-1$, the state changes by
\begin{align}\label{eq:dt_alpha_dynamics}
    x^{k+1} = (1+\alpha_i)x^{k},
\end{align}
\end{subequations}
where $\alpha_i$ is the scaling factor for $i\in [m]$. 
There is no dependence on $h$ in~\eqref{eq:dt_alpha_dynamics} as we assume the impulse occurs instantaneously.
For the parameter estimation results that follow, we use the discretized model in~\eqref{eq:dt_hybrid_dynamics}.





\section{Parameter Estimation}\label{sec:main_results}
Now we present conditions under which we can estimate the discrete-time model parameters from data.
We impose the following assumptions on the model in~\eqref{eq:dt_hybrid_dynamics}.
\begin{assumption}\label{assum:update_times}
    The sequence of $m$ updates occurs at known time indices $T_1,\dots,T_m$, where $T_0<T_1<\dots<T_m<T_{m+1}$. 
\end{assumption}

\begin{assumption}\label{assum:states}
    For all $k\in \{0\}\bigcup[T_{m+1}]$, $x^k$ and $h$ are known.
\end{assumption}

\begin{lemma}\label{lem:update_interval}
\begin{subequations}
Consider the model in \eqref{eq:dt_hybrid_dynamics} under Assumptions~\ref{assum:update_times} and~\ref{assum:states}. 
Then, $\alpha_i$, $\beta_i$, and $\gamma_i$, corresponding to any update $i\in[m]$, can be uniquely identified if and only if $T_{i+1}-T_{i}>2$,  
there exist $k_1, k_2\in \{T_{i}, T_i+1, \dots, T_{i+1}-2\}$ such that
\begin{equation}\label{eq:lemma_condition}
    x^{k_1}(1-x^{k_2})x^{k_2} \neq x^{k_2}(1-x^{k_1})x^{k_1},
\end{equation}
and 
\begin{equation}
    x^{T_i-1}\neq 0.
\end{equation}
\end{subequations}
\end{lemma}

\begin{proof}
    Since $x^k$ and $h$, for all $k\in \{T_i-1, T_i, \dots, T_{i+1}-2\}$, are known, we can rewrite~\eqref{eq:dt_hybrid_dynamics}, for one update $i\in[m]$, as
    \begin{align}\label{eq:discrete_spread_process_matrix}
        \begin{bmatrix}
            x^{T_{i}} - x^{T_{i}-1} \\
            \vdots \\
            x^{T_{i+1}-1}-x^{T_{i+1}-2}
        \end{bmatrix} = \Phi_i \begin{bmatrix}
            \alpha_{i} \\
            \beta_{i} \\
            \gamma_{i} 
        \end{bmatrix},
    \end{align}
    where, from~\eqref{eq:dt_sis_dynamics} and~\eqref{eq:dt_alpha_dynamics}, 
    \begin{align}\label{eq:phi_n_def}
        \Phi_i = \begin{bmatrix}
            x^{T_i-1} & 0 & 0 \\
            0 & h(1-x^{T_{i}})x^{T_{i}} & -hx^{T_{i}} \\
            \vdots & \vdots & \vdots \\
            0 & h(1-x^{T_{i+1}-2})x^{T_{i+1}-2} & -hx^{T_{i+1}-2}
        \end{bmatrix}.
    \end{align}
    We can rewrite $\Phi_i$ as the block diagonal matrix
    \begin{align}\label{eq:phi_n_block_def}
        \Phi_i = \begin{bmatrix}
            \Phi^0_i & \mathbf{0}  \\
            \mathbf{0} & \Phi^1_i
        \end{bmatrix},
    \end{align}
    where 
    \begin{subequations}
    \begin{align}\label{eq:phi_n_blocks}
        \Phi_i^0 &= \begin{bmatrix}
            x^{T_i-1}
        \end{bmatrix} \\
        \Phi_i^1 &= \begin{bmatrix}
            h(1-x^{T_{i}})x^{T_{i}} & -hx^{T_{i}} \\
             \vdots & \vdots \\
            h(1-x^{T_{i+1}-2})x^{T_{i+1}-2} & -hx^{T_{i+1}-2}
        \end{bmatrix}.
    \end{align}
    \end{subequations}
    Note that from~\eqref{eq:discrete_spread_process_matrix}, we can uniquely identify the parameters  $\alpha_i$, $\beta_i$, and $\gamma_i$ if and only if $\Phi_i$ is full column rank.
    Further, the block diagonal matrix $\Phi_i$ is full column rank if and only if all block matrices are full column rank.

    If $x^{T_i-1}\neq 0$, $\Phi_i^0$ has a column rank of one, with one column.
    If $x^{T_i-1}=0$, then $\Phi_i^0$ has a column rank of zero. 
    
    Since $T_{i+1}-T_{i}>2$, $\Phi_i^1$ has at least two rows.
    By the assumption that there exists $k_1, k_2\in \{T_{i}, T_i+1, \dots, T_{i+1}-2\}$ such that~\eqref{eq:lemma_condition} holds, $\Phi_i^1$ has a column rank equal to two with two columns.
    If there do not exist $k_1, k_2\in \{T_{i}, T_i+1, \dots T_{i+1}-2\}$ such that~\eqref{eq:lemma_condition} holds, the rank is less than two. 
    Therefore, we have proved the result.

\end{proof}

If we were to assume that the state $x^k$ never dies out, that is, $x^k>0$ for all $k\in \{T_{i}-1, T_i, \dots, T_{i+1}-2\}$, then the inequality condition in~\eqref{eq:lemma_condition} 
can be relaxed:
\begin{corollary}
    Consider the model in \eqref{eq:dt_hybrid_dynamics}, under Assumptions~\ref{assum:update_times} and~\ref{assum:states}, and 
    assume $x^k>0$, for all $k\in \{T_{i}-1, T_i, \dots, T_{i+1}-2\} $. 
    Then, $\alpha_i$, $\beta_i$, and $\gamma_i$, corresponding to any update $i\in[m]$, can be uniquely identified if and only if $T_{i+1}-T_{i}>2$,  
    and there exist $k_1, k_2\in \{T_{i}, T_i+1, \dots, T_{i+1}-2\}$ such that
    \begin{equation}\label{eq:cor_condition}
        x^{k_1}\neq x^{k_2}.
    \end{equation}
\end{corollary}

\begin{proof}
Since, by assumption, $x^{k_1},x^{k_2}>0$, we can divide both sides of~\eqref{eq:lemma_condition} by~${x^{k_1}x^{k_2}}$, and after some basic rearranging we are left with~\eqref{eq:cor_condition}.
\end{proof}


Recall that, immediately prior to update $T_i$, the state $x^k$ experiences an impulse, with the temporary dynamics defined as $x^{k+1}=(1+\alpha_i) x^k$ for $k=T_i-1$.
After each update, $x^k$ evolves according to the SIS dynamics with the new pair of parameters, $\beta_i,\gamma_i$.
Note that for the first interval, there is no impulse prior to $T_0$ and thus no $\alpha_0$.
Utilizing the result from Lemma~\ref{lem:update_interval}, we have the following theorem.
\begin{theorem}\label{thm:model_estimation}
\begin{subequations}
Consider the model in \eqref{eq:dt_hybrid_dynamics}, under Assumptions~\ref{assum:update_times} and~\ref{assum:states}. 
The hybrid model parameters $\Theta=[\beta_{0}\  
        \gamma_{0} \ 
        \alpha_1 \ 
        \dots \ 
        \alpha_{m} \ 
        \beta_{m} \ 
        \gamma_{m} ]^T$
can be uniquely identified, if and only if, the following hold: \\
for $i=0$:
\begin{enumerate}
    \item $T_{1}>2$,
    \item there exist $k_1, k_2\in \{0, 1, \dots ,T_{1}-2\}$ such that
    \begin{equation}\label{eq:thm_0_ineq_assump}
        x^{k_1}(1-x^{k_2})x^{k_2} \neq x^{k_2}(1-x^{k_1})x^{k_1};
    \end{equation}
\end{enumerate}
for all $i\in[m-1]$:  
\begin{enumerate}
    \item $T_{i+1}-T_{i}>2$,
    \item there exist $k_1, k_2\in \{T_{i}, T_i+1, \dots ,T_{i+1}-2\}$ such that
    \begin{equation}\label{eq:thm_n_ineq_assump}
        x^{k_1}(1-x^{k_2})x^{k_2} \neq x^{k_2}(1-x^{k_1})x^{k_1},\text{ and}
    \end{equation}
    \item $x^{T_i-1}\neq 0$;
\end{enumerate}
and for $i=m$: 
\begin{enumerate}
    \item $T_{m+1}-T_{m}>1$,
    \item there exist $k_1, k_2\in \{T_{m}, T_m+1, \dots ,T_{m+1}-1\}$ such that
    \begin{equation}\label{eq:thm_m_ineq_assump}
        x^{k_1}(1-x^{k_2})x^{k_2} \neq x^{k_2}(1-x^{k_1})x^{k_1}, \text{ and}
    \end{equation}
    \item $x^{T_m-1}\neq 0$.
\end{enumerate}
\end{subequations}
\end{theorem}

\begin{proof}
    Since $x^k$ and $h$, for all $k\in \{0\}\bigcup[T_{m+1}]$, are known, we can rewrite~\eqref{eq:dt_hybrid_dynamics} as
    \begin{align}\label{eq:dt_hybrid_process_matrix_extended}
        \begin{bmatrix}
            x^{1} - x^{0} \\
            \vdots \\
            x^{T_{m+1}}-x^{T_{m+1}-1}
        \end{bmatrix} = \Psi \Theta
        ,
    \end{align}
    where the matrix $\Psi$ is the block diagonal matrix
    \begin{align}
        \Psi = \begin{bmatrix}
            \Phi_0 & &  \mathbf{0} \\
             & \ddots & \\
            \mathbf{0}  & & \Phi_{m}
        \end{bmatrix},
    \end{align}
    with $\Phi_0$ defined as 
    \begin{align}
        \Phi_0 &= \begin{bmatrix}
            h(1-x^{0})x^{0} & -hx^{0} \\
             \vdots & \vdots \\
            h(1-x^{T_{1}-2})x^{T_{1}-2} & -hx^{T_{1}-2}
        \end{bmatrix};
    \end{align}
    $\Phi_i$, for $i\in[m-1]$, is defined by~\eqref{eq:phi_n_block_def}; and
    $\Phi_m$ is defined as
    \begin{align}\label{eq:phi_m_block_def}
        \Phi_m = \begin{bmatrix}
            \Phi^0_m & \mathbf{0}  \\
            \mathbf{0} & \Phi^1_m
        \end{bmatrix},
    \end{align}
    where $\Phi_m^0,\Phi_m^1$ are defined as
    \begin{subequations}
    \begin{align}\label{eq:phi_m_blocks}
        \Phi_m^0 &= \begin{bmatrix}
            x^{T_m-1}
        \end{bmatrix} \\
        \Phi_m^1 &= \begin{bmatrix}
            h(1-x^{T_{m}})x^{T_{m}} & -hx^{T_{m}} \\
             \vdots & \vdots \\
            h(1-x^{T_{m+1}-1})x^{T_{m+1}-1} & -hx^{T_{m+1}-1}
        \end{bmatrix}.
    \end{align}
    \end{subequations}
    Note that from~\eqref{eq:dt_hybrid_process_matrix_extended}, we can uniquely identify the parameters in~$\Theta$ if and only if $\Psi$ is full column rank.
    Further,
    the block diagonal matrix $\Psi$ is full column rank if and only if all block matrices are full column rank.

    The matrix $\Phi_0$ will be full column rank, and thus $\beta_0,\gamma_0$ will be uniquely identifiable, by~\cite[Thm~7]{estimation2020pare}, if and only if, $T_{1}>2$ and
    there exist $k_1, k_2\in \{0, 1, \dots, T_{1}-2\}$ such that
    \begin{equation*}
        x^{k_1}(1-x^{k_2})x^{k_2} \neq x^{k_2}(1-x^{k_1})x^{k_1}.
    \end{equation*}

    By the proof of Lemma~\ref{lem:update_interval}, we know that each matrix $\Phi_i$, for $i\in[m-1]$, will be full column rank if and only if $T_{i+1}-T_{i}>2$, 
    there exist~$k_1, k_2\in \{T_{i}, T_i+1, \dots, T_{i+1}-2\}$ such that
    \begin{equation*}
        x^{k_1}(1-x^{k_2})x^{k_2} \neq x^{k_2}(1-x^{k_1})x^{k_1},
    \end{equation*}
    and that 
    \begin{equation*}
        x^{T_i-1}\neq 0.
    \end{equation*}
    




    The matrix $\Phi_m$ will be full column rank, and thus $\alpha_m$, $\beta_m$, and $\gamma_m$ will be uniquely identifiable, if and only if both $\Phi_m^0$ and $\Phi_m^1$ are full column rank.
    If $x^{T_m-1}\neq 0$, $\Phi_m^0$ has a column rank of one, with one column.
    If $x^{T_m-1}=0$, then $\Phi_m^0$ has a column rank of zero.
    Since $T_{m+1}-T_{m}>1$, $\Phi_m^1$ has at least two rows.
    By the assumption that there exists $k_1, k_2\in \{T_{m}, T_m+1, \dots, T_{m+1}-1\}$ such that~\eqref{eq:thm_m_ineq_assump} holds, $\Phi_m^1$ has a column rank equal to two with two columns.
    If there do not exist $k_1, k_2\in \{T_{m}, T_m+1, \dots T_{m+1}-1\}$ such that~\eqref{eq:thm_m_ineq_assump} holds, the rank is less than two.
    Therefore, we have proved the result.

\end{proof}
If the conditions in Theorem~\ref{thm:model_estimation} are met, we can leverage
\eqref{eq:dt_hybrid_process_matrix_extended} to solve for the model parameters.
We validate this technique using simulation data and then validate the model, leveraging this technique, using real data.



\section{User Data}\label{sec:data}


Software-as-a-service (SaaS) and live service video games, which receive regular updates and rely on recurring revenue streams, such as micro-transactions, have become dominant in the gaming industry~\cite{LELONEKKULETA2021106592,de2016analysis}.
While not all players engage in micro-transactions, having more active players generally correlates with increased revenue. 
Predicting player counts over time can help developers optimize server capacity, preventing disruptions during user spikes that can harm player experience and profitability, while reducing resource costs spent during periods of low engagement.

The dataset employed in Section~\ref{sec:simulation} consists of the peak number of concurrent users per day for the live-service game Apex Legends\footnote{\url{https://store.steampowered.com/app/1172470/}} (Apex).
Steam, the largest video game distribution platform on computers, provides a public API\footnote{\url{https://steamcommunity.com/dev}} for tracking the number of current users in any specific game. 
Although Steam does not provide past data for player counts, historical datasets are available through community-run websites.
The most well-known site is called SteamDB. 
On each game page, there is a plot with the historical player counts, including the daily peak number of concurrent players. 
A snippet of our data is shown in Figure~\ref{fig:dataset}, and the full dataset is available for download 
on Apex Legend's SteamDB page\footnote{\url{https://steamdb.info/app/1172470/charts/}}.

\begin{figure}
    \centering
    \includegraphics[width=0.98\linewidth]{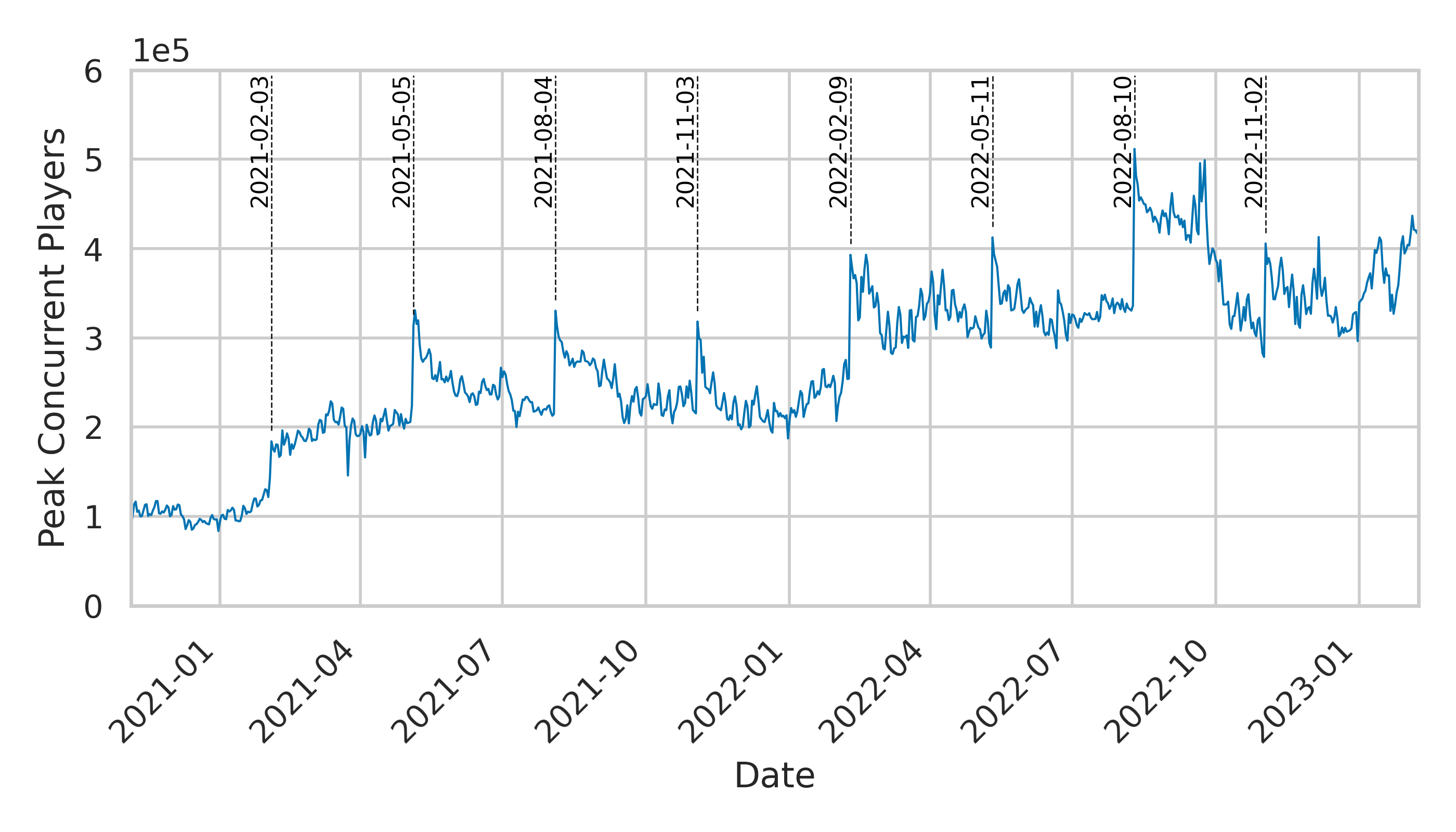}
    \caption{Player count dataset. The plot shows the peak number of concurrent players per day from November 2020, when Apex Legends was first made available through Steam, to Feb 2023. The dashed vertical lines represent the dates on which significant updates were released.}
    \label{fig:dataset}
\end{figure}

\section{Simulation}\label{sec:simulation}

The hybrid model parameters, defined in Section~\ref{sec:model}, are estimated using 
the technique from Section~\ref{sec:main_results}.
The rest of this section is structured as follows.
First, we validate our estimation technique by estimating the parameters of synthetically-generated data.
We leverage the discrete-time parameter estimation technique from Section~\ref{sec:main_results} on the continuous-time model to evaluate how the estimation technique performs.
We then apply observation noise and generate data with process noise to additionally test the robustness of the parameter estimation technique.
Finally, we validate the model on real data by estimating parameters and comparing the model fit and prediction.

\subsection{Estimation Validation}

\begin{figure}[t]
    \begin{subfigure}[t]{\linewidth}
        \centering
        \includegraphics[width=0.9\linewidth]{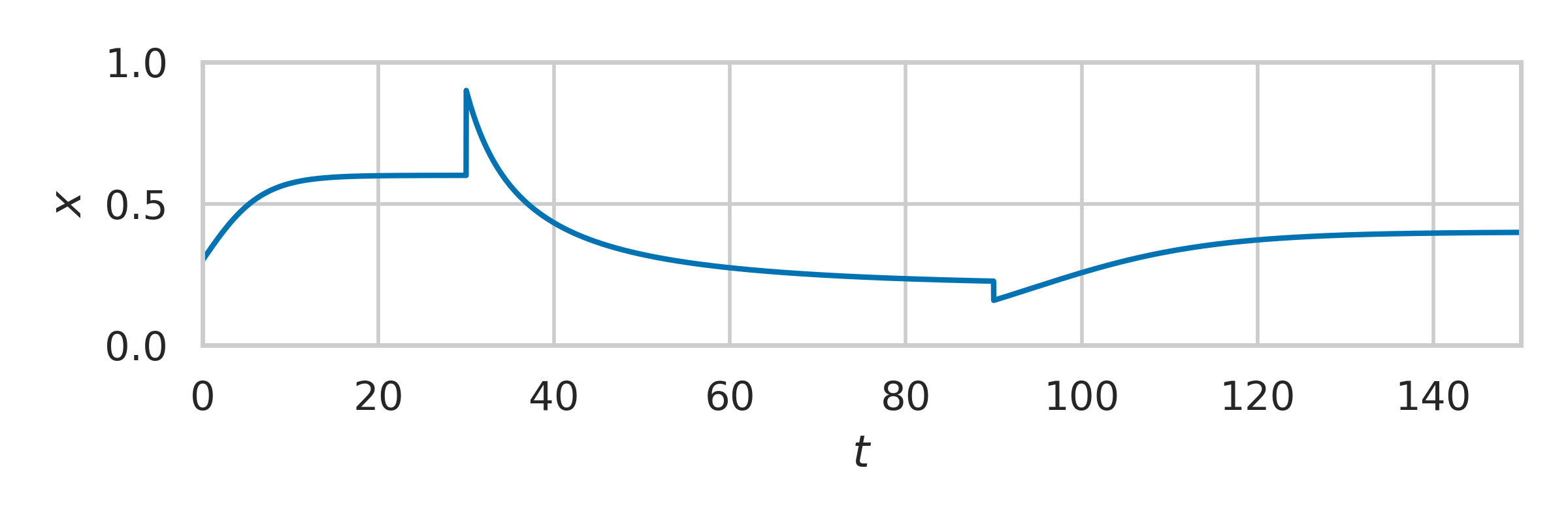}
        \vspace{-1em}
        \caption{No noise}
        \label{fig:sim_data_noiseless}
    \end{subfigure}
    \\
    \begin{subfigure}[t]{\linewidth}
        \centering
        \includegraphics[width=0.9\linewidth]{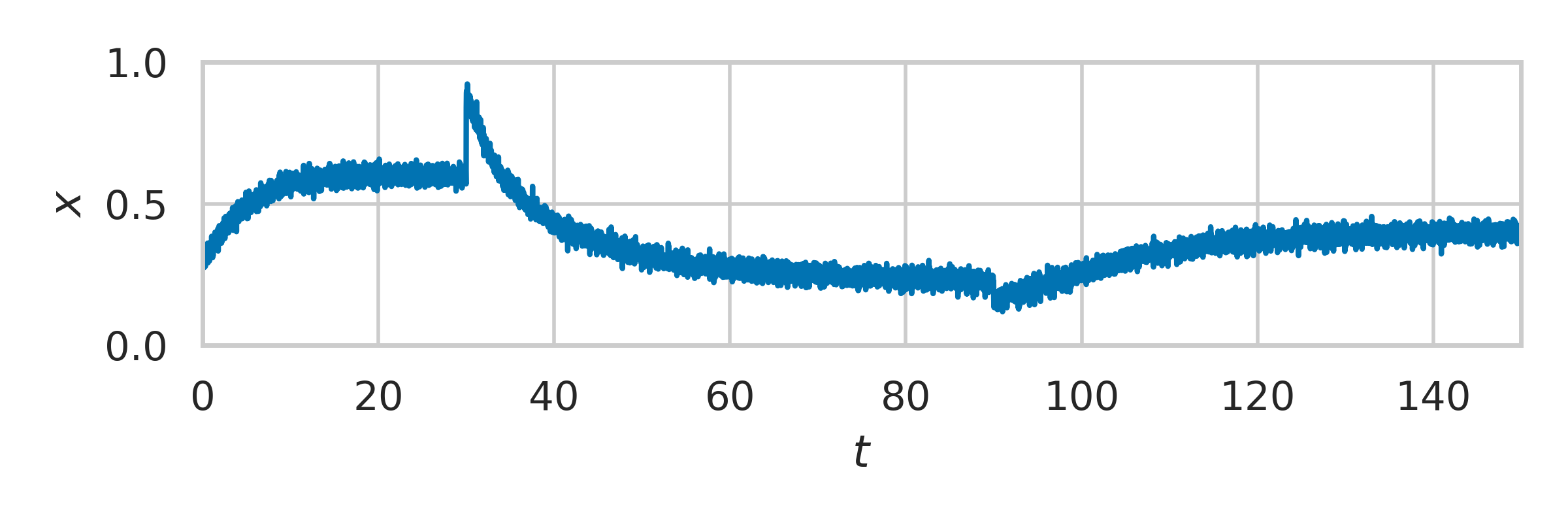}
        \vspace{-1em}
        \caption{Observation noise}
        \label{fig:sim_data_obs_noise}
    \end{subfigure}
    \\
    \begin{subfigure}[t]{\linewidth}
        \centering
        \includegraphics[width=0.9\linewidth]{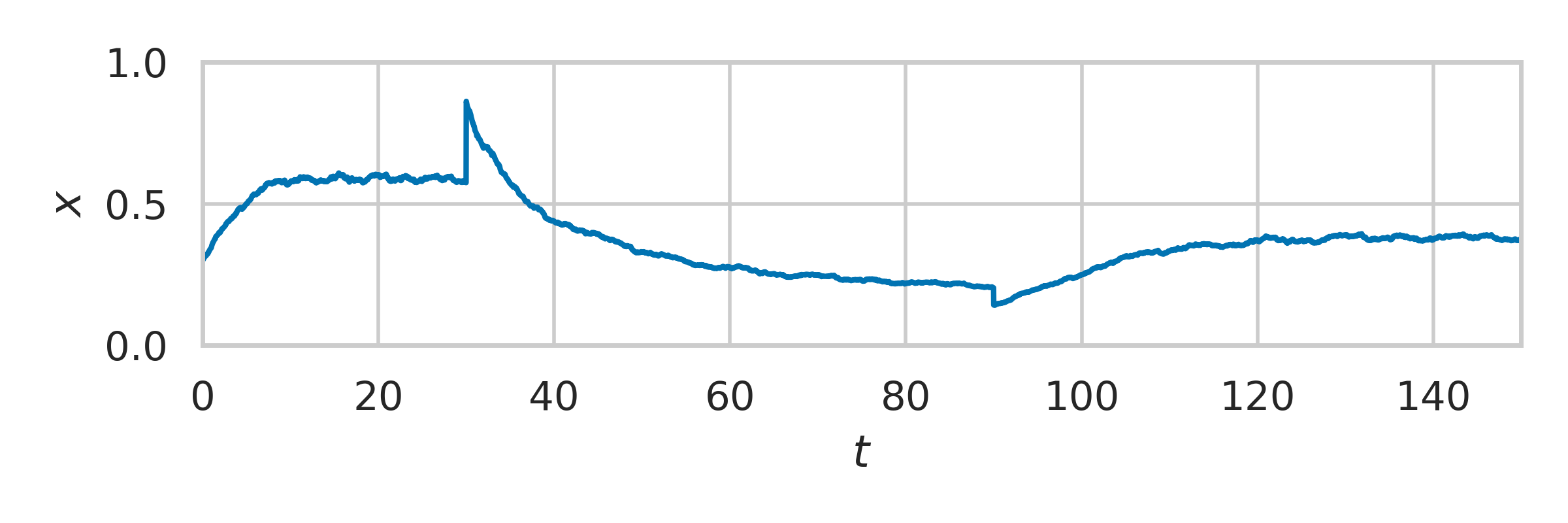}
        \vspace{-1em}
        \caption{Process noise}
        \label{fig:sim_data_proc_noise}
    \end{subfigure}
    \\
    \begin{subfigure}[t]{\linewidth}
        \centering
        \vspace{1em}
        \begin{tabular}{cccccccc}
        \toprule
        $\beta_0$ & $\gamma_0$ & $\alpha_1$ & $\beta_1$ & $\gamma_1$ & $\alpha_2$ & $\beta_2$ & $\gamma_2$ \\ 
        \midrule
        0.50 &0.20  & 0.50  &0.19 &0.15  &-0.30 &0.25 & 0.15 \\
        \bottomrule
        \end{tabular}
        \caption{Simulation parameters}
        \label{tab:sim_params}
    \end{subfigure}
    
    \caption{The simulated data from~\eqref{eq:ct_sis_dynamics}--\eqref{eq:ct_alpha_dynamics} used for validating the estimation technique. There are two update events, occurring at $t_1=30$ and $t_2=90$. The system parameters used are listed in~\ref{tab:sim_params}.
    The plot in~\ref{fig:sim_data_noiseless} is the base signal with no noise, \ref{fig:sim_data_obs_noise} is the signal with Gaussian observation noise, and~\ref{fig:sim_data_proc_noise} is the signal generated with process noise.}
    \label{fig:validation_data}
\end{figure}

\begin{figure*}[h]
    \begin{subfigure}[t]{0.3\textwidth}
        \includegraphics[width=\linewidth]{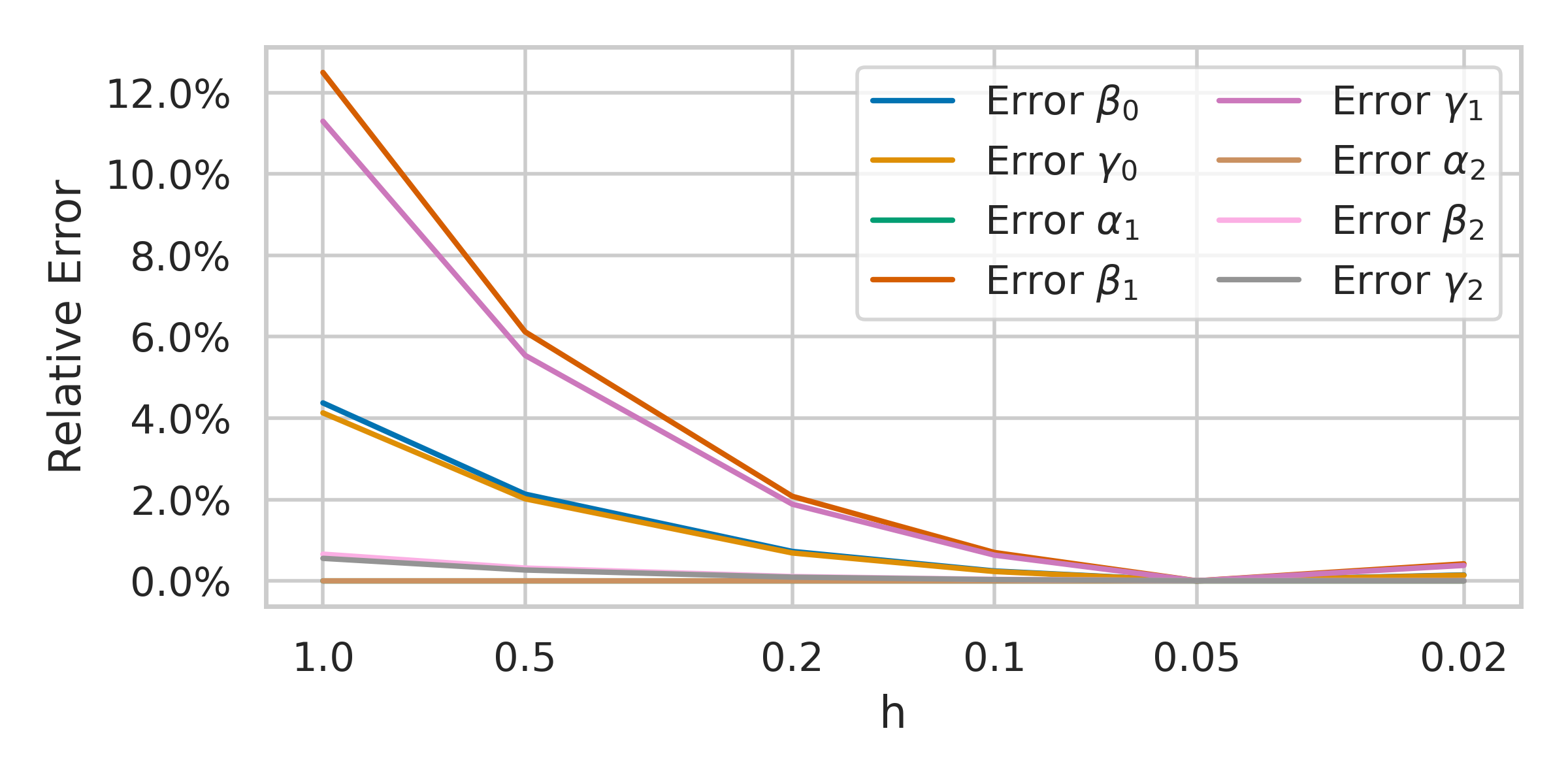}
        \caption{Parameter error with no noise}
        \label{fig:param-error-no-noise}
    \end{subfigure}
    \hfill
    \begin{subfigure}[t]{0.3\textwidth}
        \includegraphics[width=\linewidth]{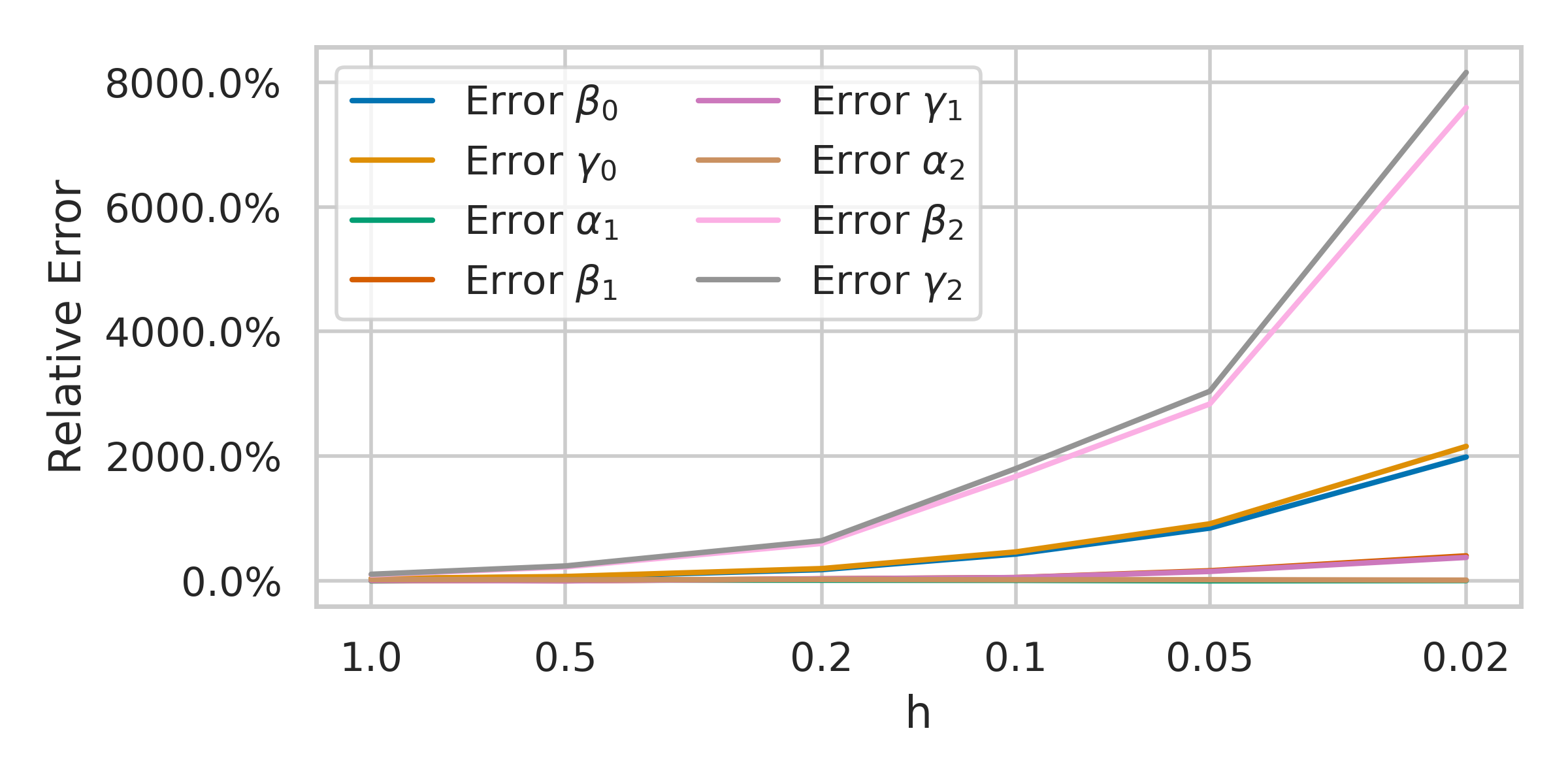}
        \caption{Parameter error with observation noise}
        \label{fig:param-error-obs-noise}
    \end{subfigure}
    \hfill
    \begin{subfigure}[t]{0.3\textwidth}
        \includegraphics[width=\linewidth]{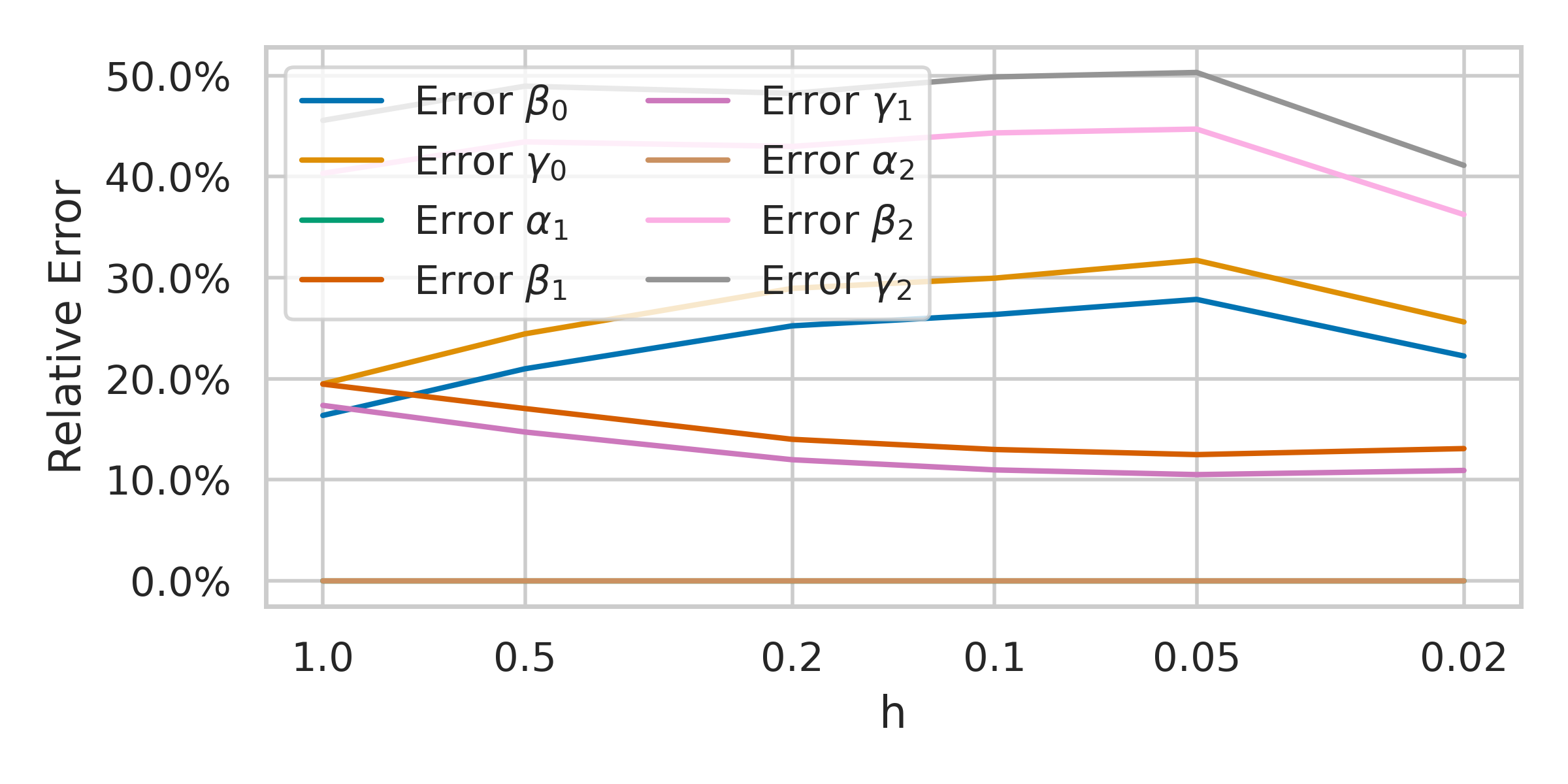}
        \caption{Parameter error with process noise}
        \label{fig:param-error-proc-noise}
    \end{subfigure}
    \\
    \begin{subfigure}[t]{0.3\textwidth}
        \includegraphics[width=\linewidth]{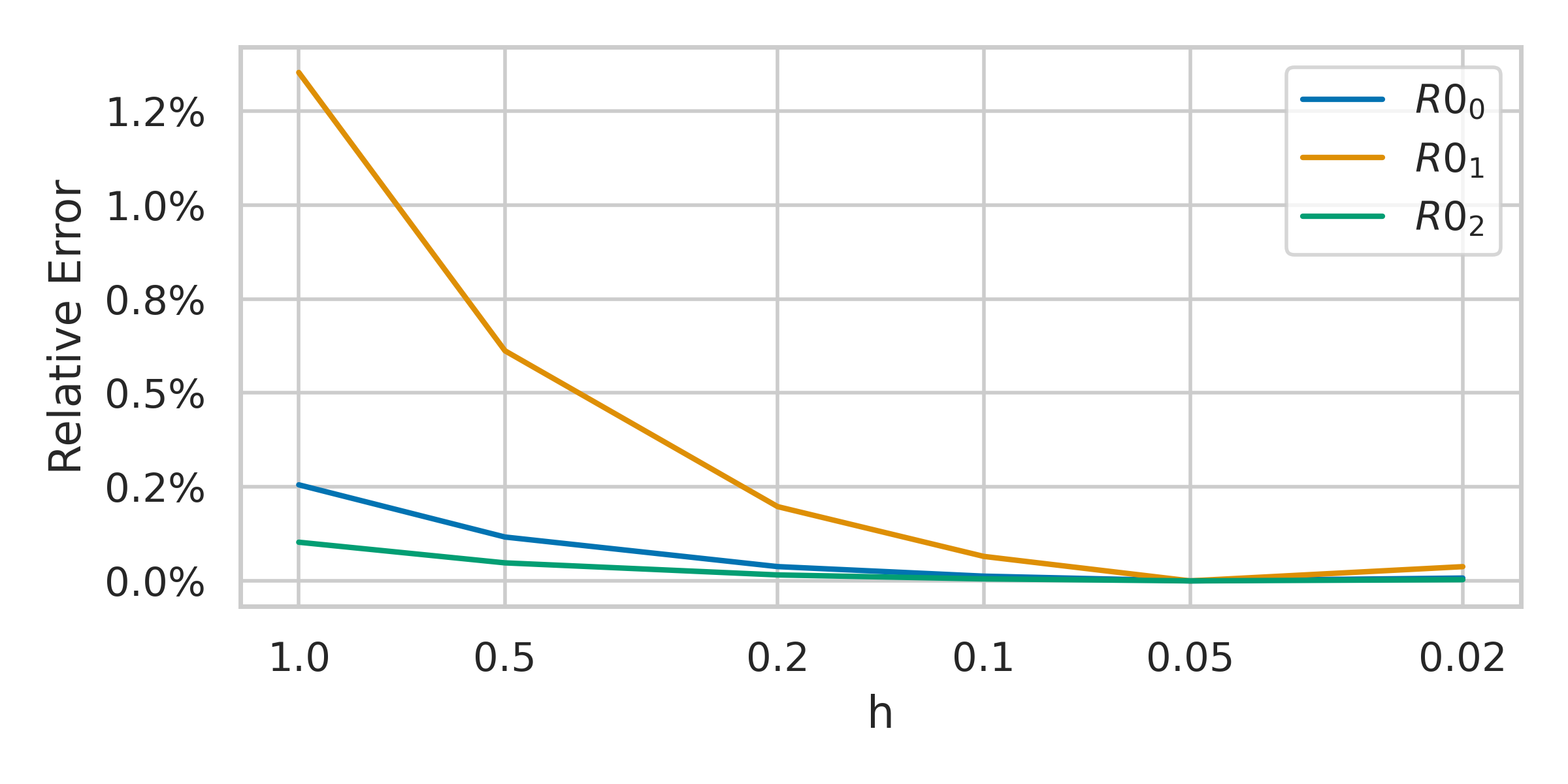}
        \caption{Reproduction number error with no noise}
        \label{fig:r0-error-no-noise}
    \end{subfigure}
    \hfill
    \begin{subfigure}[t]{0.3\textwidth}
        \includegraphics[width=\linewidth]{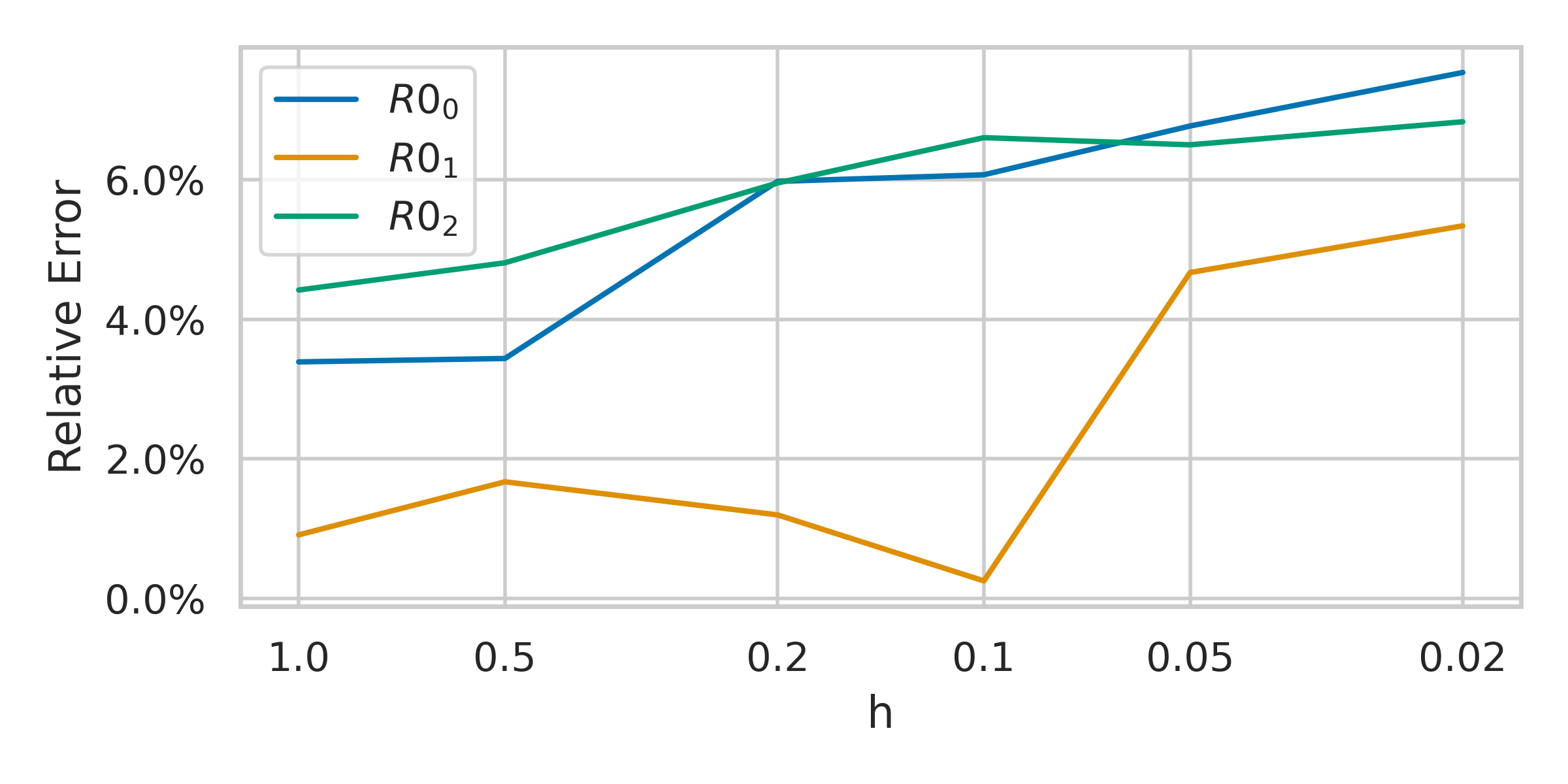}
        \caption{Reproduction number error with observation noise}
        \label{fig:r0-error-obs-noise}
    \end{subfigure}
    \hfill
    \begin{subfigure}[t]{0.3\textwidth}
        \includegraphics[width=\linewidth]{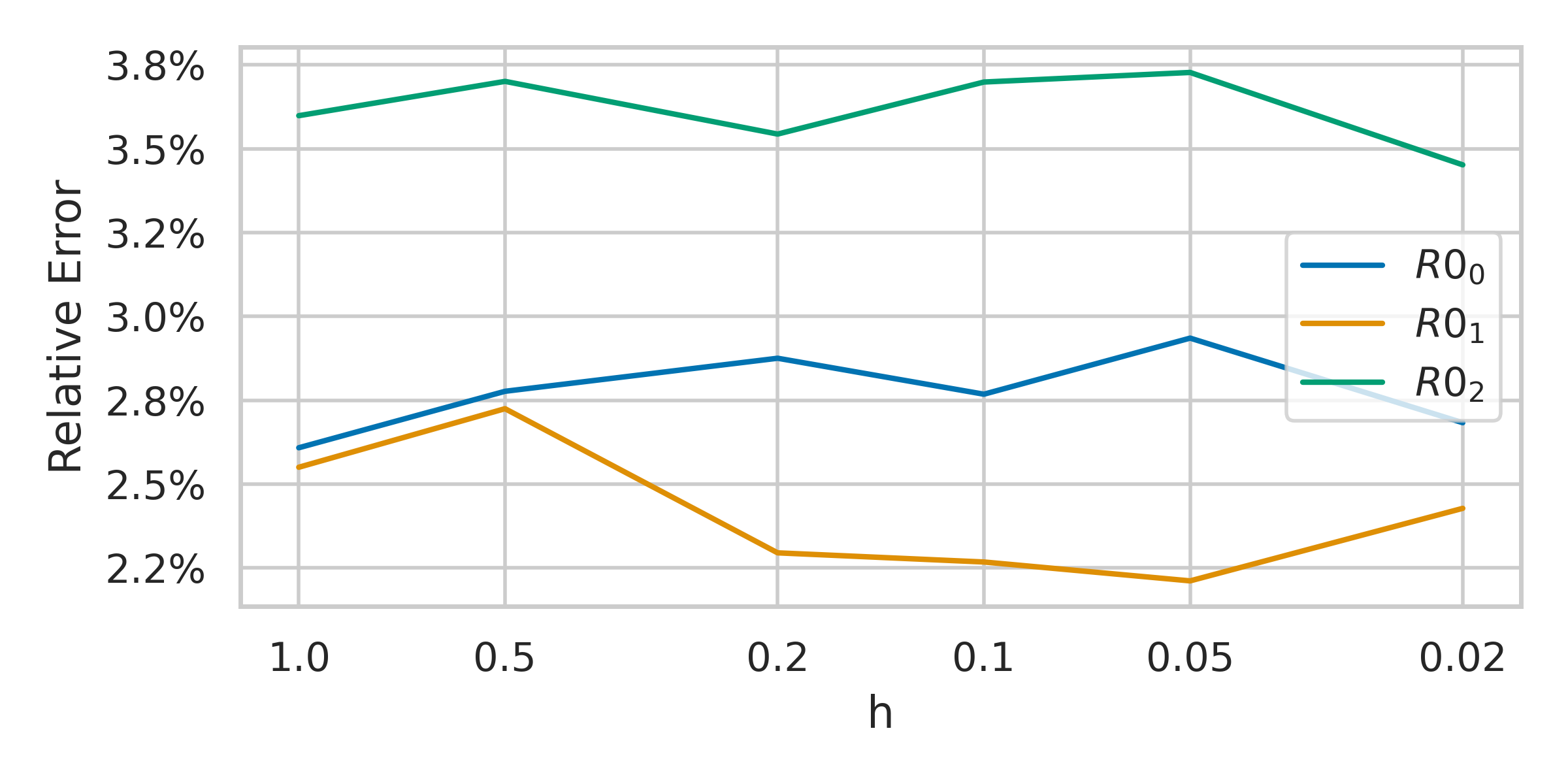}
        \caption{Reproduction number error with process noise}
        \label{fig:r0-error-proc-noise}
    \end{subfigure}

    \caption{The top row of plots shows the relative error of the estimated model parameters from the true parameters used for simulation, for each of the three simulation cases: no noise, observation noise, and process noise. The bottom row of plots shows the relative error of the reproduction number for each interval. Without any noise, the estimated parameters and reproduction numbers improve as the step size $h$ decreases. While both observation and process noise reduce the accuracy of the estimation of the true system parameters, the reproduction numbers of the estimated parameters remain close to the true reproduction number.}
    \label{fig:param_estimation_figs}
\end{figure*}

We apply our estimation technique to three cases of simulated data: (1) data generated without noise, (2) data with additive Gaussian observation noise, and (3) data with process noise.
The simulated data are generated using the continuous-time hybrid system model in~\eqref{eq:ct_sis_dynamics}--\eqref{eq:ct_alpha_dynamics}, and plots of the generated data are displayed in Figure~\ref{fig:validation_data}.
There are $m = 2$ update events, occurring at $t_1 = 30$ and $t_2 = 90$.
The simulation parameters are listed in Figure~\ref{tab:sim_params}.

For the estimation technique, we use the discrete-time approximation of the hybrid model and evaluate the estimation error for multiple step sizes $h$.
In the SIS model, the magnitudes of $\beta$ and $\gamma$ influence how quickly the state converges to equilibrium, while their ratio, defined as the reproduction number $R_0 = \beta / \gamma$, determines the value of the equilibrium.
Therefore, we present both the parameter estimation accuracy and the accuracy of the reproduction number for each update interval in Figure~\ref{fig:param_estimation_figs}.
Further details and observations relevant to the estimation for each simulated dataset are provided in the following subsections.

\subsubsection{Noiseless Data}\label{sec:sim_valid_noiseless}
We first estimate the true parameters using the noiseless simulation data.
As shown in Figures~\ref{fig:param-error-no-noise} and~\ref{fig:r0-error-no-noise}, the error in the estimated parameters and reproduction numbers decreases as the time step $h$ decreases.
This result aligns with expectations because smaller time steps provide a more accurate approximation of the continuous-time model.
For instance, when $h = 1$, the estimated reproduction number deviates by no more than $1.25\%$.
However, there is a marginal increase in error when the step size decreases from $0.05$ to $0.02$.
This can be attributed to numerical errors arising from the computation of the pseudoinverse, where the sparsity of the matrix can lead to small inaccuracies.

After confirming that the estimation method performs well under ideal conditions, we next evaluate how accurately the parameters can be estimated in the presence of noise.

\subsubsection{Observation Noise}\label{sec:sim_valid_obs_noise}
To evaluate the effect of observation noise, Gaussian noise with a standard deviation of $\sigma = 0.02$ was added to the simulation data.
The estimation error results are shown in Figures~\ref{fig:param-error-obs-noise} and~\ref{fig:r0-error-obs-noise}.
We observed that adding uncorrelated observation noise resulted in increased estimation error for both the model parameters and reproduction numbers, especially as the step size $h$ decreased.
However, the estimated reproduction number error remains below $8\%$ for all cases, with the largest error occurring at $h = 0.02$.

The estimation accuracy decreased as the time required for the state to reach the equilibrium decreased, indicating greater sensitivity to observation noise when convergence is faster.
For example, if the parameters $\beta_1$ and $\gamma_1$ are changed from $0.19$ and $0.15$, respectively, to $1.9$ and $1.5$, the reproduction number and the equilibrium value remain the same, but the estimation of $\beta_1$ and $\gamma_1$ becomes more sensitive to the noise.
This increased sensitivity occurs because fewer data points are available in the transient period when convergence is faster.
Additionally, the parameters are highly sensitive to noise when the magnitude of the update engagement spike is small relative to the observation noise.

We next evaluate the impact of process noise on parameter estimation.

\begin{figure}
    \centering
    \includegraphics[width=\linewidth]{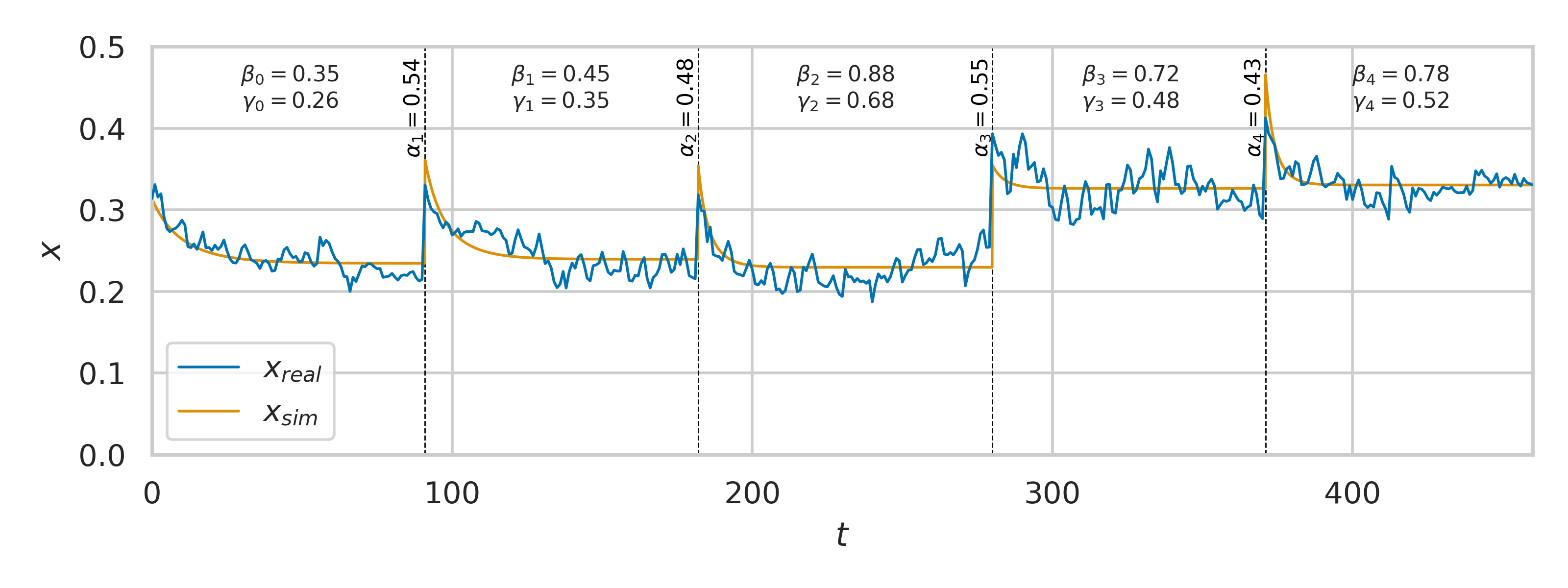}
    \vspace{-1em}
    \caption{Real data and simulation results. The blue line is the real data, and the orange line is the simulated data. The estimated parameters used to generate the simulated data are listed at the top of the plot. The overall fit of the simulated data to the real data is good, but the simulated demand spikes either over or undershoot the real demand.}
    \label{fig:real_data_sim_comparison}
\end{figure}

\subsubsection{Process Noise}\label{sec:sim_valid_proc_noise}
To assess the effect of process noise, we add noise directly to the system states and evaluate the accuracy of the parameter estimation.
The system with process noise is defined as
\begin{equation}\label{eq:process_noise}
    dx = \beta_i (1 - x) x \, dt - \gamma_i x \, dt + \sigma x \, dW,
\end{equation}
where $W$ represents the standard Wiener process.
The noise magnitude is scaled by the current state $x$, and the noise strength is set to $\sigma = 0.02$.

Our results show that the accuracy of parameter estimation with process noise did not exhibit a clear relationship with the step size $h$.
Despite the introduction of noise, the estimated parameters remained within $50\%$ of the true parameter values for all tested step sizes.
Similarly, the reproduction numbers were estimated within $4\%$ of the true values across all cases.


\subsection{Model Validation with Real Data}

After validating the estimation method and addressing practical challenges, we proceed to estimate parameters using real-world data, presented in Section~\ref{sec:data}, in order to validate our proposed model.
We estimate the parameters for multiple updates simultaneously, assuming a total population size of $N = 1{,}000{,}000$, and use these parameters to generate a fit for the data using our hybrid model.


The estimated parameters, and resulting simulation, over five intervals are visible in Figure~\ref{fig:real_data_sim_comparison}, where we compare the simulation results with the original data.
In the simulation, we used only the estimated parameters and the initial value of the data.
The results show a good fit, especially during the first three intervals.
However, starting from Interval 3, the simulated demand reaches the equilibrium more quickly than the actual data, possibly due to an underestimation of the impulse response.
Furthermore, at Update 4, the simulated impulse magnitude is significantly larger than in the real data, suggesting a potential overestimation.
These discrepancies suggest weekend/weekday patterns introduce noise in the $\alpha_i$ estimation, which relies on just two data points. 
This limited data makes $\alpha_i$ more sensitive to noise compared to the estimates for $\beta_i$ and $\gamma_i$.


Overall, the simulation provided a good fit for the data, and instills confidence that the model in~\eqref{eq:ct_sis_dynamics}--\eqref{eq:ct_alpha_dynamics} can capture real demand behavior.
However, this fit assumes all data was known, while a dynamic demand model would likely be used in an online fashion to predict future demand. 
To evaluate how well this model would perform the task of predict future demand, we next evaluate how much data is needed to predict the demand within each update interval.
The primary goal of doing these predictions is not to derive theoretical conclusions, but to identify practical considerations for using our estimation technique in real-world applications.

\section{Conclusion}\label{sec:conclusion}

We proposed a novel hybrid system model that combines continuous-time epidemic dynamics with discrete impulse events.
We presented necessary and sufficient conditions for estimating the parameters of a discrete-time approximation of the hybrid model.
We validated the estimation technique on synthetic data across multiple resolutions of discrete-time approximations, and analyzed the effects of observation and process noise on the parameter estimation technique.
Finally, we applied the estimation method to real user count data, demonstrating the model's ability to fit real-world demand patterns.

Our simulations show that the hybrid model effectively captures the dynamics of user engagement in systems influenced by discrete events, such as updates.
The estimation method exhibits robustness to noise and provides reasonable accuracy in estimating key features, in particular the reproduction number.
However, practical application to real data revealed challenges, particularly in estimating impulse parameters due to limited data points and daily fluctuations in user activity.
These findings highlight the importance of considering data variability and event timing when using the model for prediction.

Future research could focus on enhancing the estimation techniques to better handle sparse or noisy data, possibly by incorporating filtering methods or Bayesian approaches.
Extending the model to account for weekly or seasonal patterns, user heterogeneity, or incorporating control inputs could improve its predictive capabilities.
Additionally, developing methods to estimate impulse parameters more accurately, perhaps by utilizing additional data sources or machine learning techniques, could further enhance the model's utility in real-world applications.


\bibliographystyle{IEEEtran}
\bibliography{refs-acc}

\end{document}